\documentclass[a4paper,USenglish,numberwithinsect,thm-restate]{lipics-v2021}

\hideLIPIcs
\usepackage{bbm}
\pdfoutput=1
\usepackage[utf8]{inputenc}
\usepackage[T1]{fontenc}
\usepackage{amsmath,amsthm,amssymb,thmtools,booktabs,color,doi,graphicx,latexsym,url,xcolor,xspace,soul}
\usepackage{microtype}

\usepackage[ruled, lined, linesnumbered, commentsnumbered, noend, ]{algorithm2e}
\usepackage[noend]{algpseudocode}

\usepackage{sectsty}
\usepackage{todonotes}
\allsectionsfont{\boldmath}

\newcommand{\ALG}{\textsf{ALG}\xspace}
\newcommand{\OPT}{\textsf{OPT}\xspace}

\nolinenumbers

\title{Online Algorithms with Randomly Infused Advice}
\titlerunning{Online Algorithms with Randomly Infused Advice}

\author{Yuval Emek}{Technion, Israel}{}{}{}
\author{Yuval Gil}{Technion, Israel}{}{}{}
\author{Maciej Pacut}{TU Berlin, Germany}{}{}{}
\author{Stefan Schmid}{TU Berlin, Germany}{}{}{}
\authorrunning{Y. Emek, Y. Gil, M. Pacut and S. Schmid}
\Copyright{Y. Emek, Y. Gil, M. Pacut and S. Schmid}



\funding{
Research supported by the Austrian Science Fund (FWF) 
and the German Research Foundation (DFG), grant I 4800-N (ADVISE), 
2020-2023.
}

\EventEditors{Inge Li G{\o}rtz, Martin Farach-Colton, Simon J. Puglisi, and Grzegorz Herman}
\EventNoEds{4}
\EventLongTitle{31st Annual European Symposium on Algorithms (ESA 2023)}
\EventShortTitle{ESA 2023}
\EventAcronym{ESA}
\EventYear{2023}
\EventDate{September 4--6, 2023}
\EventLocation{Amsterdam, the Netherlands}
\EventLogo{}
\SeriesVolume{274}
\ArticleNo{69}
\begin{document}

\maketitle

\begin{abstract}
We introduce a~novel method for the rigorous quantitative evaluation of online
algorithms that relaxes the ``radical worst-case'' perspective of classic
competitive analysis.
In contrast to prior work, our method, referred to as randomly infused advice
(RIA), does not make any assumptions about the input sequence and does not
rely on the development of designated online algorithms.
Rather, it can be applied to existing online randomized algorithms,
introducing a means to evaluate their performance in scenarios that lie
outside the radical worst-case regime.

More concretely, an online algorithm \ALG with RIA benefits from pieces of
advice generated by an omniscient but not entirely reliable oracle.
The crux of the new method is that the advice is provided to \ALG by writing
it into the buffer $\mathcal{B}$ from which \ALG normally reads its random
bits, hence allowing us to augment it through a very simple and non-intrusive
interface.
The (un)reliability of the oracle is captured via a~parameter
$0 \leq \alpha \leq 1$
that determines the probability (per round) that the advice is successfully
infused by the oracle;
if the advice is not infused, which occurs with probability
$1 - \alpha$,
then the buffer $\mathcal{B}$ contains fresh random bits (as in the classic
online setting).

The applicability of the new RIA method is demonstrated by applying it to
three extensively studied online problems:
paging, uniform metrical task systems, and online set cover.
For these problems, we establish new upper bounds on the competitive ratio of
classic online algorithms that improve as the infusion parameter
$\alpha$ increases.
These are complemented with (often tight) lower bounds on the competitive
ratio of online algorithms with RIA for the three problems.
\end{abstract}

\thispagestyle{empty}
\setcounter{page}{0}

\section{Introduction}

\emph{Competitive ratio} is a widely used metric for evaluating the
performance of \emph{online algorithms}.
It measures the ratio between the performance of an online algorithm and that
of an optimal offline (clairvoyant) algorithm, assuming a worst-case (i.e.,
adversarial) input sequence.
Early on, it has been observed (see, e.g., \cite{Young1994loose}) that in
practice, many online algorithms outperform their theoretical worst-case
guarantees.
Indeed, in realistic scenarios, the online algorithms tend to
``enjoy a good fortune'' and rarely encounter the theoretical pitfalls that
realize the competitiveness lower bounds (cf.~\cite{KoutsoupiasP00}).

This phenomenon has led to extensive research on the analysis of online
algorithms beyond the extreme worst-case nature of traditional competitive
analysis (see \cite{KarlinK2020survey} for a recent survey).
A prominent approach in this regard is to restrict the power of the adversary
that decides on the input sequence, giving rise to the methods of
locality of reference~\cite{AlbersF18,AlbersL16,AlbersFG05},
access graph~\cite{BorodinIRS95},
smoothed analysis~\cite{ReinekeS18,BecchettiLMSV03},
random arrival order~\cite{AlbersJ21,AlbersKL21,AlbersKL21a},
independent sampling~\cite{CorreaCFOT21},
diffused adversaries~\cite{KoutsoupiasP00},
and
distributional analysis~\cite{Rivest76,Frederickson80}.
Another approach is to relax the competitive analysis definition, as done in
resource augmentation~\cite{Sleator1985},
loose competitiveness~\cite{Young1994loose},
and competitiveness with high probability \cite{KommKKM14}.
See also the surveys~\cite{DorrigivL05,BoyarIL15} for additional measures.

In this paper, we wish to advance the study of (randomized) online algorithms
beyond worst-case competitive analysis by offering a radically new point of
view on the concept of ``enjoying a good fortune'' (in terms of avoiding the
competitiveness pitfalls).
Our approach does not restrict the power of the adversary, hence we do not
need to justify any assumptions on the request sequence.
Moreover, we use the standard definition of competitive analysis (with no
relaxations).
Last but not least, in contrast to some existing ``beyond worst-case''
methods, which are limited to certain types of online problems (e.g., locality
of reference and access graph), our new method is very general and can be
applied to seemingly any online problem.

So, how do we interpret ``good fortune'' on behalf of a randomized online
algorithm \ALG{} without making any assumptions on \ALG{}'s input sequence?
The answer is simple:
we look at the outcome of \ALG{}'s random coin tosses.
That is, to make \ALG{} more fortunate, all we have to do is to increase the
chances of getting good such outcomes.

This raises another question:
what makes one outcome of \ALG{}'s random coin tosses better than another?
To answer this question, we recruit an omniscient \emph{oracle} that generates
\emph{advice} for \ALG{} in each round of the execution.
The crux of our method, called \emph{randomly infused advice (RIA)}, is that
the oracle attempts to write this advice into the buffer $\mathcal{B}$ from
which \ALG{} normally reads its random bits.
To quantitatively control \ALG{}'s good fortune, we introduce an \emph{infusion
parameter}
$0 \leq \alpha \leq 1$,
which determines the probability that the advice is (successfully) infused by
the oracle in each round (independently);
if the advice is not infused --- an event occurring with probability
$1 - \alpha$
--- then the buffer $\mathcal{B}$ contains fresh random bits (as in the classic
online setting).
Refer to Figure~\ref{fig:ria} for an illustration.

\begin{figure}[t]
\centering 
\includegraphics[width=0.6\textwidth]{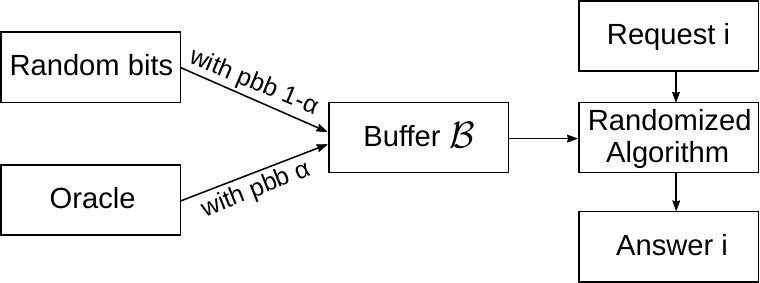}
\caption{\label{fig:ria}%
In each round, the algorithm reads its random bits from buffer $\mathcal{B}$.
Under the RIA model, the content of this buffer is replaced by the oracle's
advice for that round with probability $\alpha$, independently of other
rounds.}
\end{figure}

We emphasize that the interface between the randomized online algorithm \ALG{}
and the oracle is ``non-intrusive'', i.e., it is defined on top of the standard
computational model of (randomized) online algorithms (a.k.a.\ request-answer
games).
Therefore, the RIA method is suitable for the analysis of \textbf{existing}
online algorithms (including classic ones), facilitating the evaluation of
their performance beyond the extreme worst-case nature of traditional
competitive analysis.
This is in contrast to other advice models for online algorithms (discussed in
Section~\ref{sec:related-work}) in which the oracle-algorithm interface is
based on a designated buffer (or tape) from which the algorithm reads the
advice.
As such, these models require the development of \textbf{new}, model-specific,
algorithms and cannot be applied to existing ones.

Notice that the RIA model does not impose any limitations on the size of the
buffer~$\mathcal{B}$, and through it, on the advice size (or the number of
random bits) provided to \ALG in each round.
This raises the concern of making the online algorithm ``too powerful'' as the
(successfully) infused advice may hold  excessive information regarding the
future requests.
To overcome this concern, we restrict our attention to randomized online
algorithms which are \emph{randomness-oblivious}, namely, in each round,
\ALG has access to past requests, past answers, the current request, and the
current content of the buffer $\mathcal{B}$ (which contains the current advice
or  random bits), however \ALG cannot access the content of $\mathcal{B}$ in
previous rounds.
Indeed, all algorithms analyzed in this paper are randomness-oblivious.

The main motivation for studying the RIA method comes from analyzing the
performance of randomized online algorithms in scenarios that lie
outside the ``radical  worst-case'' regime, assumed in the classic online
computation literature.
In particular, this new method allows us to compare between different online
algorithms that exhibit the same performance guarantees in worst-case
scenarios, possibly separating between them in terms of their performance once
the scenarios get ``a little bit better'', and to do so without making any
explicit assumptions about the request sequence (or the probability
distribution thereof).

Another motivation is that the RIA model provides an abstraction for an
unreliable predictor (whose role is assumed by the oracle) whose ``mistakes''
take a~random (rather than worst-case) flavor, where the infusion parameter
$\alpha$ indicates the (expected) fraction of rounds in which the predictor is
correct.
In this regard, the non-intrusive interface between the online algorithm and
the oracle gives the RIA model a distinctive advantage over existing advice
models for online algorithms as it enables the analysis of standard online
algorithms in scenarios that include an unreliable predictor, while retaining
their worst-case guarantees.

\subsection{Our Contribution}

On top of the conceptual contribution that lies in introducing the RIA model,
we make the following technical contribution.

\paragraph*{Upper bounds.}
The applicability of the new RIA model is demonstrated on three extensively
studied online problems:
the \emph{paging} problem~\cite{Sleator1985}, for which we analyze the classic
RandomMark algorithm~\cite{FiatKLMSY91};
the uniform \emph{metrical task system (MTS)} problem~\cite{BorodinLS92}, for
which we analyze the classic UnifMTS algorithm;
and
the unweighted \emph{online set cover} problem~\cite{AlonAABN09}, for which we
analyze the influential primal-dual algorithm \cite[Ch.~4]{BuchbinderN09} with
randomized rounding (referred to as RandSC).
In all cases, our findings are similar to what is called ``smoothensss'',
``robustness'', and ``consistency'' in the literature dedicated to online
algorithms with predictions~\cite{LykourisV21,PurohitSK18}:
when augmented with RIA, the competitive ratio of these algorithms is never
worse than the original, and improves asymptotically as
$\alpha \rightarrow 1$.
Our results are cast in the following three theorems, where we denote the $k$-th
harmonic number by
$H_{k} \approx \log k$;
we emphasize that in all cases, neither the online algorithm nor the oracle
are aware of the infusion parameter $\alpha$.

\begin{theorem} \label{intro-theorem:random-mark}
The competitive ratio of RandomMark augmented with RIA with infusion parameter
$0 \leq \alpha \leq 1$
on instances of cache size $k$ is
at most
$\min \{ 2 H_{k}, \frac{2}{\alpha} \}$.
\end{theorem}

\begin{theorem} \label{intro-theorem:unif-mts}
The competitive ratio of UnifMTS augmented with RIA with infusion parameter
$0 \leq \alpha \leq 1$
on $n$-state instances is at most
$\min \{ 2 H_{n}, \frac{2}{\alpha} + 2 \}$.
\end{theorem}

\begin{theorem} \label{intro-theorem:primal-dual}
The competitive ratio of RandSC augmented with RIA with infusion parameter
$0 \leq \alpha \leq 1$
on instances with $n$ elements and maximum element degree $d$ is at most
$O (\min \{ \log d \log n, \frac{\log n}{\alpha} \})$.
\end{theorem}

\paragraph*{Lower bounds.}
On the negative side, we prove that the upper bound promised in
Theorem~\ref{intro-theorem:random-mark} is asymptotically tight for the class
of \emph{lazy} algorithms, which are not allowed to change their cache
configuration unless there is a page miss.

\begin{theorem} \label{intro-theorem:paging-lazy-lower-bound}
There does not exist a~lazy (randomness-oblivious) online paging algorithm
augmented with RIA with infusion parameter
$0 \leq \alpha \leq 1$
whose competitive ratio on instances of cache size $k$ is better than
$\min \{ H_{k}, \frac{1}{\alpha} \}$.
\end{theorem}

Omitting the restriction to lazy algorithms, we can establish a~weaker lower
bound.

\begin{theorem} \label{intro-theorem:paging-lower-bound}
There does not exist a (randomness-oblivious) online paging algorithm
augmented with RIA with infusion parameter
$0 \leq \alpha \leq 1$
whose competitive ratio on instances of cache size $k$ is better than
$\min \{ H_{k}, \frac{1}{k \cdot \alpha} \}$.
\end{theorem}

The uniform MTS problem generalizes the paging problem on instances that
include
$n = k + 1$
pages.
As Theorems \ref{intro-theorem:paging-lazy-lower-bound} and
\ref{intro-theorem:paging-lower-bound} hold (already) for such instances,
their promised lower bounds are transferred to the uniform MTS problem, where
laziness translates to online MTS algorithms that may switch state only when
the processing cost is positive~\cite{FiatM03} (an algorithm class that
includes UnifMTS).

\begin{theorem} \label{intro-theorem:unif-mts-lazy-lower-bound}
There does not exist a lazy (randomness-oblivious) online uniform MTS
algorithm augmented with RIA with infusion parameter
$0 \leq \alpha \leq 1$
whose competitive ratio on $n$-state instances is better than
$\min \{ H_{n - 1}, \frac{1}{\alpha} \}$.
\end{theorem}

\begin{theorem} \label{intro-theorem:unif-mts-lower-bound}
There does not exist a (randomness-oblivious) online uniform MTS algorithm
augmented with RIA with infusion parameter
$0 \leq \alpha \leq 1$
whose competitive ratio on $n$-state instances is better than
$\min \{ H_{n - 1}, \frac{1}{(n - 1) \cdot \alpha} \}$.
\end{theorem}

For online set cover, we establish a lower bound for lazy algorithms, namely,
online algorithms which are allowed to buy a set only if it contains the current
(uncovered) element (an algorithm class that includes RandSC).

\begin{theorem} \label{intro-theorem:set-cover-lower-bound}
There does not exist a lazy (randomness-oblivious) unweighted online set cover
algorithm augmented with RIA with infusion parameter
$0 \leq \alpha \leq 1$
whose competitive ratio on instances with maximum element degree $d$ is better
than
$\min \{ \frac{1}{2} \log d, \frac{1}{2 \cdot \alpha} \}$.
\end{theorem}

\subsection{Novelty and Additional Related Work}
\label{sec:related-work}

\paragraph*{Models of Advice.}
A well-known and suitable advice model for machine-learned predictions is the
model of online algorithms with untrusted advice introduced by Lykouris and Vassilvitskii~\cite{LykourisV21}, where the existing literature includes
papers on
paging~\cite{LykourisV21,Rohatgi20,JiangP020, BansalCKPV22},
metrical task system~\cite{AntoniadisCE0S20},
and
online set cover via the primal-dual approach~\cite{BamasMS20}.
In this model, the predictor may be faulty, and the competitive ratio depends
on its error so that for low error, the algorithm should perform close to the
offline optimum (a.k.a.\ \emph{consistency}), while even for large error, the
algorithm should still fallback to guarantees similar to those of
non-augmented online algorithms (a.k.a.\ \emph{robustness}).

Another well-known advice model is the \emph{perfect} advice model
\cite{EmekFKR11,
BockenhauerKKKM17advice}
under which many online problems have been
studied, including paging, metrical task system~\cite{BoyarFKLM2017survey},
and online set cover~\cite{DobrevEKKKKM17}.
In this model, the oracle is fully trustworthy, and its power is therefore
quantified via the size (i.e., number of bits) of the advice provided to the
online algorithm.
This model is related to lookahead~\cite{Grove95}, where an algorithm is given some number of future requests in advance.
The model of perfect advice was later extended to untrusted advice, retaining its focus on measuring the required advice size~\cite{ADJKR20}.

Unlike these two advice models, the RIA model does not require any new algorithmic
features (e.g., a designated advice tape) and is therefore applicable to
existing (standard) online algorithms.
Furthermore, our model does not limit the advice size, unlike the perfect
advice model, and still allows to arrive at asymptotically tight lower bounds
under natural assumptions, in contrast to  the machine-learned prediction
model where no general lower bounds are known.

\paragraph*{Online algorithms for paging, MTS, and set cover.}
Two optimally competitive algorithms for paging are known:
PARTITION~\cite{McGeochS91} and EQUITABLE~\cite{AchlioptasCN00}.
For the uniform MTS problem, a
$(2 H_{n})$-competitive
algorithm was presented in \cite{BorodinLS92}, later improved to
$H_{n} + O(\sqrt{\log n})$
in \cite{IraniS98};
the latter result nearly matches the $H_{n}$ lower bound of \cite{BorodinLS92}.

For online set cover, the state-of-the-art competitive ratio upper bounds are
$O (\log m \log n)$
for the weighted case
\cite{AlonAABN09}
and
$O (\log m \log (n / \OPT))$
for the unweighted case
\cite{BuchbinderN2009},
where $m$ and $n$ denote the number of sets (an upper bound on the maximum
element degree $d$) and the number of elements, respectively;
interestingly, both bounds can be realized by deterministic online
algorithms.
On the negative side, no (randomized) online algorithm has a competitive ratio
better than
$\Omega (\log m)$
\cite{Korman04}
and no deterministic online algorithm has a competitive ratio better than
$\Omega (\log m \log n / (\log \log m + \log \log n))$
\cite{AlonAABN09}.
If the (randomized) online algorithm is required to admit a polynomial time
implementation, then the competitiveness lower bound improves to
$\Omega(\log m \log n)$
assuming that
$NP \nsubseteq BPP$
\cite{Korman04}.

\section{Online Algorithms with Randomly Infused Advice}
\label{sec:model}

We begin by recalling standard definitions of online algorithms as request-answer games~\cite{Ben-DavidBKTW94}.
Our model of online algorithms with randomly infused advice is then defined as a~generalization of this model.

\subsection{Online Algorithms as Request-Answer Games}

Consider a~finite sequence
$\sigma = \langle r_{1}, \dots, r_{|\sigma|} \rangle$
of \emph{requests}, where each request $r_{i}$ is taken from a~set
$\mathcal{R}$.
A \emph{solution} for $\sigma$ is a~sequence
$\lambda = \langle a_{1}, \dots, a_{|\sigma|} \rangle$
of \emph{answers}, where each answer $a_{i}$ is taken from a~set
$\mathcal{A}$.
For a~given minimization problem, the quality of a~solution
$\lambda$ for a~request sequence $\sigma$ is determined by means of a
\emph{cost function}
$f : \mathcal{R}^{|\sigma|} \times \mathcal{A}^{|\sigma|} \rightarrow
\mathbb{R} \cup \{\infty\}$.\footnote{%
We restrict our attention to minimization problems as these are the problems
addressed in the current paper.
Extending our setting to maximization problems is straightforward.}
Let
$\OPT({\sigma}) = \inf_{\lambda \in \mathcal{A}^{|\sigma|}} f(\sigma,\lambda)$
denote the cost of an \emph{optimal solution} for $\sigma$.

In the realm of online algorithms, the requests are revealed one-by-one, in
discrete \emph{rounds}, so that upon receiving request $r_{i}$ in round $i$,
a (randomized) online algorithm $\ALG$ outputs the (random) answer
$a_{i}$ irrevocably.
That is, the solution
$\lambda_{\ALG} = \langle a_{1}, \dots, a_{|\sigma|} \rangle$
produced by $\ALG$ is defined so that each answer $a_{i}$ is computed 
as a~function of
(1)
the request subsequence
$r_{1}, \dots, r_{i}$;
(2)
the answer subsequence
$a_{1}, \dots, a_{i - 1}$;
and
(3)
round $i$'s random bit string
$\mathcal{B}_{i} \in_{R} \{ 0, 1 \}^{L}$,
where the parameter
$L \in \mathbb{Z}_{\geq 0}$
is specified by the algorithm's designer (possibly as a~function of the
parameters of the problem).\footnote{%
We use a~single parameter $L$ (that is often kept implicit in the online
algorithm's description) for simplicity of the exposition;
it can be easily generalized to a (not necessarily bounded) sequence
$L_{1}, L_{2}, \dots$
of round-dependent parameters.}

The performance of an online algorithm $\ALG$ is measured via competitive
analysis:
we say that $\ALG$ is \emph{$c$-competitive} if there exists a~constant $b$
(that may depend on the parameters of the problem) such that
$\mathbb{E}[\ALG (\sigma)] \leq c \cdot \OPT(\sigma) + b$
for any request sequence $\sigma$, where 
$\ALG(\sigma)$ is the random variable that takes on the cost of the
solution produced by $\ALG$ in response to a~request sequence $\sigma$.
The request sequence $\sigma$ is assumed to be determined by
a malicious \emph{adversary};
we stick to the convention of an \emph{oblivious adversary}~\cite[Ch. 4]{Borodin1998} which means that
the adversary knows $\ALG$'s description, but is unaware of the outcome of
$\ALG$'s random coin tosses.

\subsection{Randomly Infused Advice}
\label{sec:randomly-infused-advice}
In this paper, we introduce an extension of online algorithms, referred to as
online algorithms with \emph{randomly infused advice (RIA)}.
In the RIA model, an algorithm $\ALG$ is assisted by a~powerful, yet not
entirely reliable, \emph{oracle} that has access to the entire request
sequence~$\sigma$.
Formally, for any request sequence
$\sigma = \langle r_{1}, \dots, r_{|\sigma|} \rangle$
and round
$1 \leq i \leq |\sigma|$,
the oracle $\mathcal{O}$ is defined by an \emph{advice function}
$\mathcal{O}_{\sigma, i} : \mathcal{A}^{i - 1} \rightarrow \{ 0, 1 \}^{L}$
that maps each answer subsequence
$\langle a_{1}, \dots, a_{i-1} \rangle$
to a~bit string
$\mathcal{O}_{\sigma, i}(a_1, \dots, a_{i-1}) \in \{ 0, 1 \}^{L}$,
referred to as the round $i$'s \emph{advice}. Notice that the length of the advice
bit string is equal to the length $L$ of $\ALG$'s random bit string.

The RIA model is associated with an \emph{infusion parameter}
$0 \leq \alpha \leq 1$
that quantifies the (un)reliability of the oracle $\mathcal{O}$.
Specifically, in each round $i$, the bit string $\mathcal{B}_{i}$ (provided to
the online algorithm in that round) is now determined based on the following
random experiment (independently of the other rounds):
with probability $\alpha$, the round $i$'s advice is \emph{infused} into
$\mathcal{B}_{i}$, that is,
$\mathcal{B}_{i} \gets \mathcal{O}_{\sigma, i}(a_{1}, \dots, a_{i - 1})$;
with probability
$1 - \alpha$,
the bit string $\mathcal{B}_{i}$ is picked uniformly at random, that is,
$\mathcal{B}_{i} \in_{R} \{ 0, 1 \}^{L}$.

In other words, in each round $i$ where the infusion is successful (an event
occurring with probability $\alpha$), the oracle's advice ``smoothly''
substitutes the random bit string $\mathcal{B}_{i}$ before it is provided to
$\ALG$;
if the infusion is not successful, then $\mathcal{B}_{i}$ remains a~random bit
string.
We emphasize that $\ALG$ and $\mathcal{O}$ are not aware (at least not
directly) of whether the advice is successfully infused in the round $i$, nor
are they aware of the infusion parameter $\alpha$ itself.

The competitive ratio of online algorithms $\ALG$ with RIA is typically
expressed as a~function of the infusion parameter $\alpha$, where the extreme
case of
$\alpha = 0$
corresponds to standard online computation (with no advice).
The ultimate goal is to provide guarantees on the competitiveness of $\ALG$
for any
$0 \leq \alpha \leq 1$.

\subsection{Randomness-Oblivious Online Algorithms}
\label{ssec:randomness-oblivious}
Recall that the aforementioned definition of online algorithms dictates that
when the online algorithm $\ALG$ determines the answer $a_{i}$ associated with
round $i$, it is aware of the requests $r_{i'}$ and answers $a_{i'}$
associated with past rounds
$i' < i$,
as well as the request $r_{i}$ and random bit string $\mathcal{B}_{i}$
associated with the current round $i$, however it is not aware (at least not
directly) of the random bit strings $\mathcal{B}_{i'}$ associated with past
rounds
$i' < i$.
This model choice is made to prevent an online algorithm $\ALG$ with RIA
from passing information received through the (successfully infused) advice to
future rounds, thus over-exploiting the lack of an explicit (model specific)
bound on the length of the random / advice bit strings.
To distinguish the online algorithms that adhere to this formulation from
general online algorithms (that may maintain a~persistent memory that encodes
past random bits), we refer to the former as \emph{randomness-oblivious}
online algorithms.

Without further restrictions, any general online algorithm $\ALG$ can be
simulated by a randomness-oblivious online algorithm $\ALG'$.
Indeed, during round $i$, algorithm $\ALG'$ can compute, for each answer
$a \in \mathcal{A}$,
the probability that $\ALG$ selects answer $a$ in response to the current
request $r_{i}$ conditioned on having
$(r_{1}, \dots, r_{i - 1})$
and
$(a_{1}, \dots, a_{i - 1})$
as the past requests and answers, respectively;
since
$(r_{1}, \dots, r_{i - 1})$,
$(a_{1}, \dots, a_{i - 1})$,
and
$r_{i}$
are known to $\ALG'$ during round $i$, this computation is feasible (though
potentially tedious).

For online algorithms with RIA on the other hand, access to past coin tosses
provides the algorithm with excessive power.
To demonstrate this phenomenon, consider some problem $\mathcal{P}$ with a
finite request set $\mathcal{R}$ and a (non-randomness-oblivious) online
algorithm $\ALG^{*}$ for $\mathcal{P}$ designed so that the random bit string
$\mathcal{B}_{i}$ provided to $\ALG^{*}$ in round $i$ encodes a ``guess'' of
the entire suffix
$(r_{i + 1}, r_{i + 2}, \dots)$
of future requests;
algorithm $\ALG^{*}$ then operates by adopting the guess $\mathcal{B}_{i'}$
provided in the earliest round $i'$ such that $\mathcal{B}_{i'}$ is consistent
with all the requests revealed to $\ALG^{*}$ so far (this is well defined since
the guess $\mathcal{B}_{i}$ is always consistent during round $i$) and
selecting an answer which is optimal under the assumption that
$\mathcal{B}_{i'}$ remains consistent subsequently.

When provided with RIA, online algorithm $\ALG^{*}$ obtains a correct guess
$\mathcal{B}_{i}$ of the entire future request suffix
$(r_{i + 1}, r_{i + 2}, \dots)$
once the advice is successfully infused (the
earliest round in which this happens is a geometric random variable with
parameter $\alpha$).
Following that, $\ALG^{*}$ may ``toggle through'' at most
$i - 1$
incorrect guesses
$\mathcal{B}_{i'}$,
$1 \leq i' < i$,
before it adopts the correct guess $\mathcal{B}_{i}$ and acts optimally
subsequently.
For many problems $\mathcal{P}$ (including, e.g., paging and any MTS with a
finite set of tasks), this would allow us to bound the expected cost incurred
by $\ALG^{*}$ before it ``converges to the optimal behavior'' as a function of
$\alpha$ and the problem parameters, thus charging this cost to the additive
constant.
Ultimately, this implies that the competitive ratio of $\ALG^{*}$ is $1$ for
any constant
$\alpha > 0$,
which clearly defies the smoothness expected from RIA as an analytical tool.

\section{Paging}
\label{sec:pagining}

In the \emph{online paging} problem~\cite{Sleator1985}, we manage a~two-level memory hierarchy, consisting of a~slow memory that stores the set of all $n$ pages, and a~fast memory, called the \emph{cache}, that stores any size $k$ subset of pages.
We are given a~sequence $\sigma$ of requests to the pages.
If a~requested page is not in the cache, a~\emph{page fault} occurs, and the page must be moved to the cache. Since the cache is limited in size, we must specify which page to evict to make space for the requested page.
The goal is to minimize the number of page faults.

In this section, we analyze an elegant randomized online algorithm RandomMark, introduced by Fiat, Karp, Luby, McGeoch, Sleator and Young~\cite{FiatKLMSY91}, in the randomly infused advice framework.
The algorithm RandomMark maintains a~bit associated with each page in the cache.
Initially the bits of all pages are set to 0 (the pages are \emph{unmarked}), and after requesting a~page, we bring it to the cache if it is not in the cache yet, and we set its bit to 1 (we \emph{mark} the page).
To bring a~page to the cache, we may need to evict another page to make space for it.
In such a~case, RandomMark evicts a~page uniformly at random chosen from the unmarked pages.
If no unmarked page exists, we unmark all pages.
This strategy has been shown to be $2H_k$-competitive~\cite{FiatKLMSY91}, where $H_k$ is the harmonic number, and no randomized algorithm can be better than $H_k$-competitive.

\subsection{RandomMark With Infused Advice}

With help of randomness, the classic RandomMark decides on the final candidate to evict: a~random node among unmarked pages.
With infused advice, in some rounds the randomness source used by RandomMark contains advice instead of random bits.
The presence of clairvoyent advice brings obvious advantages, but also brings challenges: not all pages can be evicted, only the unmarked ones.

\paragraph*{Unmarked Longest-Forward-Distance Oracle.}
An optimal offline algorithm for paging is to evict the item with the access time furthest in the future~\cite{Belady66}, also known as \emph{longest forward distance} (LFD) algorithm.
However, we cannot directly design an oracle for RandomMark around LFD, as it may advise to evict a~marked page, but RandomMark never evicts marked pages.
Hence, we propose a variant of this algorithm that can act as an oracle for RandomMark. Such an oracle, denoted $O_{ULFD}$, advises RandomMark to evict the page with the longest forward distance \emph{among the unmarked items} of RandomMark.

\paragraph*{Analysis of RandomMark.}
How well can RandomMark perform with infused advice? 
To find out, we consider the RandomMark algorithm assisted with the oracle $O_{ULFD}$, and we express the algorithm's competitive ratio of in terms of the infusion parameter $\alpha$ (the probability of receiving advice in each round).
Later in this paper, we will show that RandomMark with $O_{ULFD}$ is asymptotically optimal (Theorem~\ref{thm:lblazy}).

\begin{restatable}{theorem}{randomMarkThm}
The competitive ratio of RandomMark with the oracle~$O_{ULFD}$ with RIA on instances of cache size $k$ (against
the oblivious adversary) is at most
$\min \{ 2 H_{k}, \frac{2}{\alpha} \}$,
where
$H_k$ is the $k$-th harmonic number, and
$0 \leq \alpha \leq 1$
is the infusion parameter.
\label{thm:rmark-general}
\end{restatable}
Before proving this theorem, we recall the definition of a $k$-phase partitioning of an input sequence, and we derive sufficient conditions to stop incurring further page faults in a~phase.

We begin by recalling basic definitions from the analysis of RandomMark by~\cite{FiatKLMSY91}. 
We consider the \emph{$k$-phase partition} of the input sequence $\sigma$, following the notation from~\cite{Borodin1998}: phase 0 is the empty sequence, and each phase $i>0$ is the maximal sequence following the phase $i-1$ that contains at most $k$ distinct page requests since the start of the $i$th phase.
In a~phase of any marking algorithm, a~page requested in the phase is \emph{stale} if it is unmarked but was marked in the previous phase, and a~page is \emph{clean} if it is neither stale nor marked.

In addition to these standard definitions, we define the set of \emph{vanishing pages} as the set of the pages requested in the previous phase, but not in the current phase.
We claim that after evicting all vanishing pages, marking algorithms incur no further cost in the phase, since a~configuration is reached where all the remaining requests in the current phase are free (page hits).

\begin{lemma} \label{lem:evicting-Q-ends}
Fix an input sequence $\sigma$, consider its $k$-phase partition, and fix any
phase $P$ that is not the first or the last phase.
Then,
(1)
we have exactly $c$ vanishing pages, where $c$ is the number of clean pages in
the phase;
and
(2)
after evicting all vanishing pages, no marking algorithm for paging incurs
further cost in the phase.
\end{lemma}
\begin{proof}
    In the phase $P$, we have exactly $k$ requests to distinct pages: to $k-c$ stale pages and to $c$~clean pages. Only the clean pages can replace the vanishing pages, hence we have exactly $c$ vanishing pages. Hence, the first claim holds.
    
    If at any point all $c$~vanishing pages are evicted, this means that all $c$ clean pages were requested in the phase already. The remaining requests in the phase can concern only stale pages. As no vanishing pages remain in the cache, the cache consists of $c$ clean pages and $k-c$ stale pages.
    Hence, after evicting all vanishing pages, any marking algorithm incurs no further cost in the phase, and the second claim holds.
  \end{proof}

Finally, we prove our main claim for paging: RandomMark is $\min \{ 2 H_{k}, \frac{2}{\alpha} \}$-competitive.
We repeat the classic arguments of~\cite{FiatKLMSY91} to arrive at the bound $2 H_k$, and we analyze the offline algorithm \emph{unmarked longest forward distance}, employed by the oracle that probabilistically interacts with the oracle, to arrive at the bound $\frac{2}{\alpha}$.

\begin{proof}[Proof of Theorem~\ref{thm:rmark-general}]
Fix any input sequence $\sigma$ and consider its $k$-phase partition.
Consider any phase that is not the first or the last one.
Let $c$ be the number of clean pages in the phase.
   
    We claim that the expected number of page faults is upper bounded by $c/\alpha$.
    If the algorithm incurs a~page fault, and it receives the oracle's advice, and there are still some vanishing pages in the cache, then the algorithm evicts a~vanishing page; this follows since the vanishing pages are not requested in the current phase, hence they have larger forward distance than other stale pages, and the vanishing pages are unmarked.
    By Lemma~\ref{lem:evicting-Q-ends}, evicting all vanishing pages means that no further cost is incurred throughout the phase, hence the number of page faults in the phase is upper bounded by the number of page faults until the algorithm receives $c$ rounds of advice from the oracle (not necessarily consecutive).
    The expected number of page faults until receiving $c$ rounds of advice is $c/\alpha$, since this is the expected number of independent tosses of $\alpha$-biased coin until getting $c$ heads outcomes.

    Next, we repeat the classic arguments of~\cite{FiatKLMSY91}: the expected number of page faults of the algorithm is also upper bounded by $c\cdot H_k$.
    Consider an $i$-th request to a~stale page in the phase for $i=1,2,3,\ldots,s$.
    Let $c(i)$ denote the number of clean pages requested in the phase immediately before the $i$-th request to a~stale page, and let $S(i)$ denote the set stale pages that remain in the cache before the $i$-th request to a~stale page, and let $s(i)=|S(i)|$.
    For $i=1,2,3,\ldots,s$, we compute the expected cost of the $i$-th request to a~stale page.
    When the algorithm serves the $i$-th request to a~stale page, exactly $s(i)-c(i)$ of the $s(i)$ stale pages are in the cache.
    The stale pages are in the cache with equal probability, say $p$, since these are never evicted with the help of advice, but are evicted uniformly from unmarked pages when a~page fault occurs in rounds without advice.
    The vanishing pages are in the cache with probability at most $p$, since they can be evicted both in the rounds with and without the advice.
    For the all $s(i)$ stale pages the probability of being in the cache sums to $1$, hence $p \le 1/s(i)$.
    Fix a~request to a~stale page.
    The page is in the cache with probability $(s(i)-c(i))\cdot p$, hence the expected cost of the request is
    \[
        1-(s(i)-c(i))\cdot p \le 1 - \frac{s(i)-c(i)}{s(i)} = \frac{c(i)}{s(i)}\le \frac{c}{k-i+1}.
    \]
    Hence, the total cost of the request to the stale pages is $\sum_{i=1}^s c / (k-i+1) \le \sum_{i=2}^k c/i = c\cdot (H_k-1)$.
    The total cost in the phase includes the cost of serving the clean page and stale pages, in total $c\cdot H_k$.

    We conclude that the number of page faults of the algorithm in a~phase is upper-bounded by both $c\cdot H_k$ and $c/\alpha$.
    By arguments of~\cite[Theorem 1]{FiatKLMSY91}, the amortized number of faults made by \OPT during the phase is at least $c/2$.
    Summing over all phases but the first and the last one, the competitive ratio is at most $\min\{2H_k, \frac{2}{\alpha}\}$. The first and the last phase incurs cost bounded by $2k$, which we account in the additive in the competitive ratio.
\end{proof}

The above analysis is asymptotically tight with the lower bound given in Theorem~\ref{thm:lb_paging}.
However, for the special case $n=k+1$, the result is tight: the competitive ratio of RandomMark with the oracle $O_{ULFD}$ is $\min\{H_k, \frac{1}{\alpha}\}$, since in each phase but the last phase, any offline algorithm pays at least $1$, and the number of clean pages is also $1$.

The algorithm RandomMark with perfect advice ($\alpha = 1$) is equivalent to an offline algorithm that evicts the unmarked item with the longest forward distance.
The Theorem~\ref{thm:rmark-general} implies that this algorithm is optimal for $n=k+1$, and a~$2$-approximation for any $n$.

\section{Uniform Metrical Task System} 
\label{sec:mts}

In the \emph{metrical task system (MTS)} problem, we are given a~finite metric space $(S,d)$ consisting of a~set $S=\{s_1,\dots ,s_n\}$ of $n$ \emph{states} and a~distance function $d:S^2\rightarrow \mathbb{R}_{\geq 0}$ assumed to be a~metric.
A \emph{task} $r\in \mathbb{R}_{\geq 0}^{n}$ is an $n$-sized vector of non-negative \emph{processing costs}, where the entry $r(i)$ is defined to be the processing cost of serving $r$ in state $s_i$.
Given a~sequence $\sigma= r^1,\dots ,r^{|\sigma|}$ of tasks, the cost of
a~schedule $s^1,\dots ,s^{|\sigma|}$ is the sum between the total transition
cost and the total processing cost. The goal in the MTS problem is to find
a~schedule of minimal cost.
We focus on algorithms for the MTS problem in the
online setting, where the state $s^i$ that serves task $r^i$ is chosen without
knowing the subsequence $r^{i+1},\dots,r^{|\sigma|}$.

In this section, we focus on the  MTS problem on a~uniform metric, i.e., the metric where $d(s_i,s_j)=1$ for all $i\neq j$. We shall present a~randomized algorithm, henceforth referred to as UnifMTS, with advice. This algorithm is inspired by the classical $2H_n$-competitive algorithm by~\cite{BorodinLS92}. 

Consider a~sequence $\sigma= r^1,\dots ,r^{|\sigma|}$ of tasks given at times $t=1,\dots,|\sigma|$. For an integer $i\in \{1,\dots,|\sigma|\}$, and $i\leq \ell<\ell'\leq i+1$, let us define the processing cost $\pi(s_j,\ell,\ell')$ of being in state $s_j$ in the time interval $[\ell,\ell']$ as $\pi(s_j,\ell,\ell')=(\ell'-\ell)\cdot r^{i}(j)$. We now naturally extend this notion to time intervals $[\ell,\ell']$ such that $i\leq \ell\leq i+1<\ell '\leq |\sigma|+1$ by defining $\pi(s_j,\ell,\ell')=\pi(s_j,\ell,i+1)+\pi(s_j,\lfloor\ell'\rfloor,\ell')+\sum_{k=i+1}^{\lfloor\ell'\rfloor-1}\pi(s_j,k,k+1)$.

We define a~partition of $[1,|\sigma|+1]$ into time intervals $[t_0,t_1],[t_1,t_2],\dots,[t_{m-1},t_m]\subseteq [1,|\sigma|+1]$ called \emph{phases} such that $t_0=1$ and $t_m=|\sigma|+1$. The $i$-th phase starts at time $t_{i-1}$. We say that a~state $s_j$ is \emph{saturated} for phase $i$ at time $t>t_{i-1}$ if the processing cost associated with being in $s_j$ during the entire time interval $[t_{i-1},t]$ is at least $1$. The $i$-th phase ends in time $t_i$, defined to be the minimal time in which all states are saturated for the $i$-th phase. Observe that upon the arrival of a~task $r^{i}$ at time $i$, an online algorithm can determine which states will become saturated for the current phase by time $i+1$. 

The UnifMTS algorithm operates as follows. Consider the task $r^i$ arriving at time $i$ and let $\varphi$ be the current phase. If the current state does not become saturated for $\varphi$ at time $i+1$, then UnifMTS stays in the same state. Otherwise, if $\varphi$ ends by time $i+1$, then UnifMTS moves to a~state that minimizes the processing cost in $r^i$. Otherwise, UnifMTS moves uniformly at random to a~state that is unsaturated for $\varphi$ at time $i+1$ (such a state exists since in this case $\varphi$ does not end by time $i+1$). We note that while phases may end at non-discrete times, the scheduling decisions made by the algorithm all occur at discrete times.   

Consider an oracle $O_{LTS}$ that advises UnifMTS to move to a~state with the longest time until saturation for the current phase. In the following theorem, we bound the competitive ratio of UnifMTS with $O_{LTS}$.

\begin{theorem}
	The competitive ratio of UnifMTS with the oracle $O_{LTS}$ against an oblivious adversary is at most $\min \{2H_n,\frac{2}{\alpha} +2\}$, where $0\leq \alpha\leq 1$ is the infusion parameter.
    \label{thm:unifmts}
\end{theorem}

\begin{proof}
Observe that an optimal offline algorithm $\OPT$ must incur a~cost of at least
$1$ during each phase. Indeed, if $\OPT$ changed states during a~phase, then
it pays at least $1$ in transition cost.
Otherwise, $\OPT$ resided in a~state that became saturated in this phase,
hence it pays a~processing cost of $1$.

We now bound the expected cost of UnifMTS with $O_{LTS}$ during a~phase
$\varphi=[t_{start},t_{end}]$. Observe that if there exists $i\in \{1,\dots
,|\sigma|\}$ such that $i\leq t_{start}< t_{end}\leq i+1$, then by definition,
at time $i$ UnifMTS moved to a~state that minimizes the processing cost
incurred during $[i,i+1]$. This means that UnifMTS pays $1$ in processing cost
during $[ t_{start}, t_{end}]$ and possibly $1$ in transition cost at time
$i$.
Thus, in this case the expected cost of UnifMTS
during $\varphi$ is at most $2$.
	
	Now we consider the case that there exists $i\in \{1,\dots ,|\sigma|\}$ such that $t_{start}<i< t_{end}$.
    We show that the cost of UnifMTS during $\varphi$ is at most $\frac{2}{\alpha} +2$.
    Let $s^{*}$ be the state given in the advice of $O_{LTS}$ during $\varphi$. 
    By definition, by the time $s^{*}$ is saturated for $\varphi$, all other states have also been saturated. Therefore, when UnifMTS receives an advice from the oracle, it transitions into the final state of phase $\varphi$. 
    Hence, the additional cost incurred by the UnifMTS in $\varphi$ following the advice is at most $2$ ($1$ for the transition to $s^{*}$ and at most $1$ for processing cost).
    Since the algorithm uses randomization only at transition rounds, hence the expected number of transitions before the algorithm receives the advice is $1/\alpha$ (recall that at each transition the advice is given with probability $\alpha$).
    For each state that we visit, we pay $1$ in transition cost.
    Since UnifMTS only moves to states that are unsaturated for $\varphi$, it pays at  most $1$ in processing cost at each state.
    Overall, the expected total cost is at most $\frac{2}{\alpha} +2$.   
	
	We now show that the cost of UnifMTS during $\varphi$ is at most $2H_n$. Notice that for every transition, UnifMTS pays $1$ in transition cost and at most $1$ in processing cost. Thus, it suffices to show that the expected number of transitions during $\varphi$ is at most $H_n$. Let $f(k)$ be the expected number of transitions UnifMTS performs given that there are $k$ unsaturated states left. Clearly, $f(1)=1$. For $k<1$, after a~single transition we have $k-1$ unsaturated states with probability at most $1/k$. Thus, $f(k)\leq f(k-1)+1/k$, which implies that $f(n)\leq H_n$. Summing over the costs of all phases, we get a~competitive ratio of $\min \{2H_n,\frac{2}{\alpha} +2\}$. 
\end{proof}

\section{Set Cover} 
\label{sec:set-cover}
In the \emph{set cover} problem, we are given a universe $\mathcal{U}$ of $n$ elements and a set $\mathcal{F}=\{S_{1},\dots,S_{m}\}$ of $m$ subsets $S_{1},\dots,S_{m}\subseteq \mathcal{U}$ such that $S_{1}\cup\dots \cup S_{m}=\mathcal{U}$. For each element $e\in \mathcal{U}$, let $\mathcal{F}(e)=\{S\in \mathcal{F}\mid e\in S\}$ be the collection of sets that cover it. In the online setting, a subset  $\mathcal{U'}\subseteq \mathcal{U}$ of elements arrive one by one in an arbitrary order.\footnote{While our results in the current section are expressed in terms of the size of the universe $n$, it can be modified to obtain the same asymptotic bounds in terms of the length of the element sequence $|\mathcal{U}'|$.} Upon the arrival of an element $e$, the algorithm is required to cover it (i.e., if $e$ was not previously covered by the algorithm, then the algorithm must select a set from $\mathcal{F}(e)$). We emphasize that the algorithm does not know $\mathcal{U'}$ (or its size) in advance and that any previously selected set cannot be removed from the solution obtained by the online algorithm. The cost of a solution to the set cover problem is the number of sets selected.

In the standard linear program (LP) relaxation for set cover, each set $S\in \mathcal{F}$ is associated with a variable $x_{S}$. The objective is to minimize the sum $\sum_{S\in \mathcal{F}}x_{S}$ subject to the constraints $\sum_{S\in \mathcal{F}(e)}x_{S}\geq 1$ for each element $e\in \mathcal{U'}$, and $x_{S}\geq 0$ for all $S\in\mathcal{F}$.

Recall that in the context of set cover in the RIA model, we focus on \emph{lazy} algorithms, i.e., algorithms that adhere to the following restrictions upon the arrival of an elemnt $e$: (1) if $e$ is already covered by the algorithm, then in the current round the algorithm does not select any additional sets to its solution; and (2) if $e$ is not covered yet, then in the current round the algorithm may only select sets from $\mathcal{F}(e)$.  Notice that this restriction prevents the trivial oracle strategy of simply advising to select all the sets of an optimal set cover at each round.

We describe an online algorithm with RIA for set cover in three stages.
First, we present an algorithm that obtains a fractional solution $\mathbf{x}$
to the relaxed LP.
Then, we present an online randomized rounding scheme that can be incorporated
into the fractional set cover algorithm to obtain an integral solution which
is feasible with high probability.
Finally, we present the oracle's advice. 

\paragraph*{Fractional set cover algorithm.} We use the basic discrete algorithm presented by Buchbinder and Naor in \cite[Chapter 4.2, Algorithm 1]{BuchbinderN09}.\footnote{We note that the algorithm presented in \cite{BuchbinderN09} is designed for weighted set cover. The algorithm presented in this paper is its application for the case of unit weights.} The algorithm operates as follows. Initially, set $x_{S}=0$ for all $S\in\mathcal{F}$. Upon arrival of an element $e$, if $\sum_{S\in \mathcal{F}(e)}x_S<1$, then update $x_{S}\leftarrow 2\cdot x_{S} +1/|\mathcal{F}(e)|$ for all $S\in\mathcal{F}(e)$. Observe that at the end of the round, it is guaranteed that the fractional primal solution maintained by the algorithm satisfies the constraint since the algorithm adds at least $1/|\mathcal{F}(e)|$ to the variable $x_{S}$ for each set $S\in \mathcal{F}(e)$.

Let $d=\max_{e\in \mathcal{U'}} |\mathcal{F}(e)|$ be the maximum degree of an element. The following assertion on the competitive ratio is established by Buchbinder and Naor in \cite{BuchbinderN09}.
\begin{lemma}[\cite{BuchbinderN09}]\label{lemma:fractional}
	The fractional set cover algorithm is $O(\log d)$-competitive.
\end{lemma}

\paragraph*{Randomized rounding.} 
An online rounding scheme that randomly obtains an integral solution from the fractional set cover algorithm was constructed by Alon et al.\ in \cite{AlonAABN06}. The solution produced by the rounding scheme of \cite{AlonAABN06} is feasible with high probability while incurring a multiplicative factor of $O(\log n)$ to the expected cost. However, this rounding method does not fit our advice framework. This is because all random coins are tossed in the beginning to compute a threshold for each set. Thus, we present a slightly different rounding method that fits our framework while maintaining similar guarantees. 

The rounding procedure operates as follows. Consider an element $e$ and let $\mathbf{x}$ and $\mathbf{x}_{int}$ be the solution maintained by the fractional algorithm and the (integral) solution maintained by the rounding scheme, respectively, at the time of $e$'s arrival. If $e$ is already covered by either the current fractional solution or the current integral solution produced by the rounding, then we do nothing (we will later show that the feasibility of $\mathbf{x}_{int}$ is maintained with high probability in this case). Otherwise ($e$ is not covered by both solutions), we update $\mathbf{x}$ according to the fractional algorithm. For each $S\in \mathcal{F}(e)$, let $x^{beg}_{S}$ be the value of the variable $x_{S}$ at the beginning of the round and let  $\delta(S)=x^{beg}_{S}+1/|\mathcal{F}(e)|$ be the additive increase to $x_{S}$ that occurs during the round. The rounding is obtained by independently selecting each set $S\in \mathcal{F}(e)$ to the cover with probability $\min\{1,\delta(S)\cdot \Theta(\log n)\}$. 

We refer to the randomized algorithm described above (i.e., the fractional set cover algorithm combined with the rounding scheme) as RandSC. The properties of RandSC are described in the following lemma.

\begin{lemma}\label{lemma:rand-sc}
RandSC is $O(\log n\log d)$-competitive and computes a feasible solution with high probability.\footnote{For simplicity, RandSC is described as a Monte Carlo algorithm. It can be easily transformed into a Las Vegas algorithm as follows: whenever an element $e$ is not covered by RandSC upon the end of a round, select an arbitrary set that covers $e$ into the solution. Notice that the added expected cost is negligible.}
\end{lemma}
\begin{proof}
	Let $\mathbf{x}$ be the solution obtained by the fractional algorithm at termination. Recall that in each round, set $S$ is selected with probability at most $\delta(S)\cdot c\log n$ (for a constant $c>0$). By linearity of expectation, the total expected cost associated with $S$ is $ O(\log n)\cdot x_{S}$. Thus, the expected cost of RandSC is $O (\log n)\cdot \sum_{S\in\mathcal{F}}x_{S}=O(\log n\log d)\cdot\OPT$.
	
	We now bound the probability that there exists an element that was not covered by the integral solution produced by RandSC when it arrived. Consider an element $e'$ arriving at round $r$. Notice that by construction, $e'$ must be covered by the fractional solution at the end of round $r$. We argue that this implies that $e'$ is covered by the integral solution with high probability. Let $\ell=|\mathcal{F}(e')|$ and let $S^{1},\dots S^{\ell}$ denote the sets in $\mathcal{F}(e')$. Let us denote by $\delta_{i,j}$ the increase to the variable $x_{S^{i}}$ associated with set $S^{i}$ in round $j$ and let $p_{i,j}$ the probability that $S^{i}$ was selected to the integral solution at round $j$. If $p_{i,j}=1$ for some $i\leq \ell$ and $j\leq r$, then $e'$ is covered by the end of round $r$ with probability $1$. Otherwise, due to the independence of selection events, the probability that $e'$ is not covered by the integral solution at the end of round $r$ is
	$$\prod_{i=1}^{\ell}\prod_{j=1}^{r}(1-p_{i,j})\leq e^{-\sum_{i=1}^{\ell}\sum_{j=1}^{r}p_{i,j}}=e^{-c\log n\sum_{i=1}^{\ell}\sum_{j=1}^{r}\delta_{i,j}}\leq n^{-c},$$
	where the final inequality holds because the fractional algorithm guarantees that $e'$ is covered at round $r$ and thus $\sum_{i=1}^{\ell}\sum_{j=1}^{r}\delta_{i,j}\geq 1$. By a union bound argument, the probability that there exists a set that is not covered by the integral solution is at most $n^{1-c}$. Thus, RandSC produces a feasible solution with probability at least $1-1/n^{c-1}$.   
\end{proof}

\paragraph*{Oracle's advice.} 
The idea of the oracle's advice is to boost the probability of selecting "good" sets while not losing the probabilistic feasibility guarantee of Lemma \ref{lemma:rand-sc}. For the sake of analysis, let us assume that the oracle is randomized (observe that this assumption does not enhance the oracle's power since the oracle can deterministically compute an optimal realization of the randomized selection). Let $\mathcal{A^{*}}\subseteq \mathcal{F}$ be an optimal solution for the set cover instance. Consider the arrival of an element $e$ that was not covered yet by both the fractional and integral solutions and let $p_{S}$ be the probability that set $S$ is selected in the current round of RandSC for each set $S\in\mathcal{F}(e)$. The oracle's advice is as follows: (1) each set $S\in \mathcal{F}(e)\cap\mathcal{A^{*}}$ is selected to the advice; and (2) each set $S\in \mathcal{F}(e)-\mathcal{A^{*}}$ is independently selected to the advice with probability $p_{S}$. Notice that the argument used in Lemma \ref{lemma:rand-sc} regarding the feasibility of the solution still holds since the oracle does not decrease the selection probability of any set at a given round.
Denoting this oracle by $O_{boost}$, we can establish the following theorem.

\begin{theorem}\label{theorem:randsc}
	The competitive ratio of RandSC with the oracle $O_{boost}$ against an oblivious adversary is $O(\log n)\cdot\min \{1/\alpha,\log d\}$, where $0\leq \alpha\leq 1$ is the infusion parameter.
\end{theorem}
\begin{proof}
	We start by showing that RandSC with $O_{boost}$ is $O(\log n\log d)$-competitive. Notice that by Lemma \ref{lemma:rand-sc}, the total expected cost associated with sets $S\in \mathcal{F}-\mathcal{A^{*}}$ is $O(\log n \log d)\cdot \OPT$. In addition, the total cost of sets in $\mathcal{A^{*}}$ is bounded by $|A^{*}|=\OPT$. Therefore, the expected cost of the solution produced by RandSC with $O_{boost}$ is $O(\log n\log d)\cdot \OPT$.
	
	We now show that RandSC with $O_{boost}$ is $O(\frac{\log  n}{\alpha})$-competitive. Consider the run of RandSC with $O_{boost}$ on some element sequence. We refer to a round as a selection round if there exists a set that is selected with a positive probability in that round. Notice that we can bound the cost of RandSC with $O_{boost}$ only in selection rounds (for non-selection rounds no cost is incurred). Observe that in each selection round, the probability of selecting a set from $\mathcal{A^{*}}$ is at least $\alpha$ (the probability of receiving advice). Moreover, if at some point in the execution all sets from $\mathcal{A^{*}}$ were selected, then there are no selection rounds after that point (since $\mathcal{A^{*}}$ covers all elements). Hence, the expected number of selection rounds during the execution is at most $|\mathcal{A^{*}}|/\alpha$.
	
	To complete our analysis, we argue that the expected cost associated with sets that are not in $\mathcal{A^{*}}$ at each selection round is $O(\log n)$. Consider a selection round in which an elements $e$ arrived. Recall that for each set $S\in \mathcal{F}(e)- \mathcal{A^{*}}$, we define $\delta(S)=x^{beg}_{S}+1/|\mathcal{F}(e)|$, where $x^{beg}_{S}$ is the value of variable $x_{S}$ at the beginning of the round, and select each set $S \in \mathcal{F}(e)- \mathcal{A^{*}}$ to the cover with probability $\min\{1,\delta(S)\cdot \Theta(\log n)\}$. Thus, the total expected cost that comes from the sets $S\in \mathcal{F}(e)- \mathcal{A^{*}}$ in the round is bounded by $O(\log n)\cdot \sum_{S\in \mathcal{F}(e)- \mathcal{A^{*}}}x^{beg}_{S}+\frac{1}{|F(e)|}\leq O(\log n) \cdot 2=O(\log n)$. Since the total cost associated with sets from  $\mathcal{A^{*}}$ is at most $|\mathcal{A^{*}}|$, we get that the total expected cost of RandSC with $O_{boost}$ is $O(\log n)\cdot |\mathcal{A^{*}}|/\alpha=O(\frac{\log  n}{\alpha})\cdot \OPT$.
\end{proof}

\section{Lower Bounds}
\label{sec:lower-bound}

In this section we show fundamental limitations of online algorithms with RIA.
First, we give a lower bound for competitiveness with RIA for online set cover, under the assumption that the algorithm is lazy (buys sets only when they are needed to cover the current element).
Second, we give a lower bound for competitiveness with RIA for paging, that we improve to an asymptotically tight lower bound for the case of lazy algorithms.
The lower bound for paging implies the lower bound for the uniform metrical task system.

\subsection{Online Set Cover}

We give a lower bound for the competitive ratio of any online randomized algorithm with RIA for online set cover. The construction of the input sequence is similar to the lower bounds given in~\cite[Theorem 2.2.1]{Korman04} and \cite[Lemma 4.6]{BuchbinderN09}.
The bound is given for randomness-oblivious (defined in Section~\ref{ssec:randomness-oblivious}) and lazy algorithms (lazy algorithms 
are allowed to buy a set only if it contains the current
element).

\begin{theorem}
    Assume that an online randomized algorithm with RIA for online set-cover is lazy, randomness-oblivious and strictly $c$-competitive against the oblivious adversary.
    Then $c \ge \min \{ \frac{1}{2}\log n, \frac{1}{2\alpha} \}$, where $n$ is the size of the universe of element, and $\alpha$ is the infusion parameter.
    \label{thm:lb_setn}
\end{theorem}
\begin{proof}
    Fix any lazy, randomness-oblivious online randomized algorithm \ALG with RIA, its oracle $O$ and the infusion parameter $\alpha$.
    The adversary is oblivious to random choices of the algorithm, but it has access to the description of the algorithm, the oracle and the infusion parameter, hence can maintain the probability distribution of \ALG's cache configurations.

    Consider a complete binary tree with $d$ leaves. The items to be covered are the nodes of the tree, and the sets are the $d$ root-leaf paths.
    Our sequence $\sigma$ will be the items on one root-leaf path, starting from the root and going downward. 

    We chose the sequence of items to request corresponding to a path in the complete binary tree as follows.
    Let $F(e)$ be the family of sets that cover the item $e$, and let $p_S$ be the probability that \ALG currently has the set $S$ in the solution.
    The first request is to the root of the tree.
    For the $i$-th request, we choose one of the children, $x$ or $y$ of the item requested in the $(i-1)$-th request, depending on the probability distribution of the sets that cover these items.
    To decide between $x$ and $y$, we choose the item $r \in \{ x, y \}$ with no smaller sum of the probability mass $\sum_{F(x)} p_S$.
   
    We consider two cases depending on whether or not the algorithm received advice for $\sigma$.

    \begin{enumerate}
        \item 
        Assume the algorithm did not receive advice for $\sigma$.
        In such case, the algorithm acts as an online algorithm without advice.
        Notice that the total probability mass of sets that do not appear in subsequent iterations add up to at least $1/2$.
        Each path has length $\log n$, and the algorithm pays at least $\frac{1}{2}$ for each such round, hence overall the algorithm pays $\frac{1}{2}\log d$.

        \item Assume the algorithm received advice for $\sigma$.
            In expectation, the number of rounds before getting advice is $\frac{1}{\alpha}$, and the algorithm pays at least $\frac{1}{2}$ for each such round, hence in total the algorithm pays at least $\frac{1}{2} \cdot \frac{1}{\alpha} = \frac{1}{2 \alpha}$.
    \end{enumerate}

    Note that $\sigma$ can be covered by a single set, namely the one that corresponds to the leaf where the path ends, hence $\OPT(\sigma) = 1$.
    The online algorithm pays at least $\min \{ \frac{1}{2}\log d, \frac{1}{\alpha} \}$ for any sequence $\sigma$ of the form described above, hence \ALG is at least strictly $\min \{ \frac{1}{2}\log d, \frac{1}{2\alpha} \}$-competitive.
\end{proof}

For lazy algorithms, we can obtain a lower bound in terms of the number of $d$.
We say that an online algorithm for online set cover is \emph{lazy} if it buys a set only if the current element is not yet covered, and then it may buy only sets that cover the current element.
The next bound is stronger than the previous one, as it the bound is on the competitive ratio in the classic sense, with the possible additive constant, as opposed to the previous bound on the strict competitiveness.

\begin{theorem}
    Assume that an online randomized algorithm with RIA for online set-cover is lazy, randomness-oblivious and $c$-competitive against the oblivious adversary.
    Then $c \ge \min \{ \frac{1}{2}\log d, \frac{1}{2\alpha} \}$, where $d$ is the maximum element degree, and $\alpha$ is the infusion parameter.
    \label{thm:lb_setcover}
\end{theorem}
\begin{proof}
    We repeat the construction from the previous proof of Theorem~\ref{thm:lb_setn} in phases, in each phase using a binary tree of $2d$ items.

    As the algorithm is lazy, it cannot buy sets from future phases, and the sets used in different phases are disjoint, hence advice received in any phase cannot decrease the cost of the algorithm in any future phase.
    
    Fix any phase.
    We consider two cases depending on whether or not the algorithm received advice in this phase.
    If the algorithm received advice, then it pays at least $\frac{1}{2} \cdot \frac{1}{\alpha} = \frac{1}{2 \alpha}$, as the expected number of rounds in this phase before receiving advice concerning sets in this phase is $\frac{1}{\alpha}$.
    Otherwise, if the algorithm did not receive advice, then it pays at least $\frac{1}{2} \cdot \log d$, following the arguments from the previous proof.
    
    In total, the algorithm pays at least $\min \{ \frac{1}{2}\log d, \frac{1}{2\alpha} \}$ in each phase, and an optimal algorithm can cover the items in each phase using a single set, hence the algorithm is at least $\min \{ \frac{1}{2}\log d, \frac{1}{2\alpha} \}$-competitive.

    Note that we can repeat this construction arbitrary number of iterations to obtain a lower bound on the competitive ratio, as opposed to a lower bound on strict competitive ratio.
    In each iteration, we use a new set of items and sets corresponding to a binary balanced tree, and the maximum number of sets that cover any item $d$ does not increase by repeating the construction. Hence, no randomized algorithm with infused advice can be better than $\min \{ \frac{1}{2}\log d, \frac{1}{2\alpha} \}$-competitive.
\end{proof}

\subsection{Paging and Metrical Task Systems}

In this section we give a lower bound for competitiveness of randomized online
algorithms with RIA for paging.
The uniform metrical task system problem generalizes the paging problem
on instances that include
$n = k + 1$
pages, hence the lower bound for paging is a common lower bound for paging and
uniform metrical task system.
We restrict our attention to randomness-oblivious algorithms, as defined in
Section~\ref{ssec:randomness-oblivious}.
Our lower bound for any randomness-oblivious algorithm is loose by a~factor of
$1/k$;
but with the~natural assumption that the algorithm is \emph{lazy}, we get rid
of the $1/k$ factor, and for lazy algorithms the upper bounds for paging
(Theorem~\ref{thm:rmark-general}) and uniform metrical task systems
(Theorem 11 in the full version \cite{full-version}) are asymptotically optimal.

To show the lower bound in this section, we apply Yao's Minimax Principle~\cite{Yao77} to competitiveness of randomized online algorithms.
In the case of classic online algorithms, the lower bound for the competitiveness of the best \emph{deterministic} online algorithm on a distribution of inputs implies a lower bound on the competitiveness of any randomized online algorithm on any input sequence.

We define a deterministic equivalent of algorithms with RIA.
To this end, we add to each request the information whether the request is served by a deterministic online algorithm or by the oracle.
We will analyze performance of such an algorithm on a distribution of requests, where each round is served by the algorithm with probability $1-\alpha$, and by the oracle with probability $\alpha$.
To give a lower bound for randomness-oblivious algorithms (as defined in
Section~\ref{ssec:randomness-oblivious}), we need to define a deterministic
equivalent of such algorithms that we refer to as \emph{deterministic
advice-oblivious} algorithms: the answer for each request not served by the
oracle is determined by the current request, previous requests and previous
answers.

To apply Yao's priciple to competitiveness of randomized online algorithms with RIA, we construct a matrix representation of the game, where the row player corresponds to a deterministic advice-oblivious algorithm combined with the offline oracle algorithm, and the column player represents the adversary who specifies the input sequence. The value in each row-column pair of the matrix equals the expected cost incurred by the algorithm-oracle pair on the input sequence, divided by the cost of an optimal offline solution for the input sequence. The choice of whether the online algorithm or the oracle serves a request is beyond the control of both the adversary and the online algorithm, and to compute the value for a row-column pair we take the expectation over all possibilities where for each request independently, the deterministic algorithm serves the request with probability $1-\alpha$, and the oracle serves the request with probability $\alpha$.
Notably, a randomness-oblivious algorithm is no more powerful than a distribution over the deterministic advice-oblivious algorithms.

\begin{theorem}
    Assume that an online randomized algorithm with RIA for online paging is randomness-oblivious and $c$-competitive against the oblivious adversary.
    Then $c \ge \min \{ H_k, \frac{1}{k \cdot \alpha} \}$, where $\alpha$ is the infusion parameter.
    \label{thm:lb_paging}
\end{theorem}
\begin{proof}
    To prove the theorem, we apply Yao's Minimax Principle~\cite{Yao77} to competitiveness of randomized algorithms.
    Consider any deterministic advice-oblivious algorithm $A$ for paging, and construct the following distribution over input sequences.
    Each round is served by $A$ with probability $1-\alpha$, and by the oracle with probability $\alpha$.
    The distribution over requests to pages is constructed as follows.
    Let $S = \{ p_1, p_2, p_3, \ldots, p_{k+1}\}$ be a~set of $k+1$ pages.
    We construct a~probability distribution for choosing a~request sequence.
    The first request $\sigma(1)$ is chosen uniformly at random from $S$.
    Every other request $\sigma(t)$, $t > 1$, is made to a~page that is chosen uniformly at random from $S \setminus \{ \sigma(t-1) \}$.
    A phase starting with $\sigma(i)$ ends with $\sigma(j)$, where $j, j > i$ is the smallest integer such that $\{ \sigma(i), \sigma(i+1), \ldots, \sigma(j) \}$ contains $k+1$ distinct pages.

    We claim that for any advice-oblivious algorithm, the advice received in past phases cannot reduce the cost of the algorithm in future phases. 
    We argue as follows.
    First, the advice-oblivious algorithm is forbidden to store past advice in its internal memory for future use. Second, no algorithm can store meaningful advice for the future in its cache configuration: each phase contains requests to $k+1$ different items, so for any cache configuration at the start of the phase, there is always at least $1$ \emph{clean page}: a page that is requested in the phase that the algorithm does not have in the cache at the start of the phase.
    
    In our bounds, we use that the average cost of the algorithm for each request is $1/k$; this follows because the requested page is random and each of its pages is outside the cache with equal probability.

    We lower-bound the cost of the algorithm in each phase in two ways, depending on whether or not the algorithm receives advice in any round of the phase.
   \begin{enumerate}
        \item 
        Assume that the algorithm does not receive advice in any round of the phase.
        In such case, the algorithm acts as an online algorithm without advice throughout the phase, and
        the expected cost of the algorithm in the phase is at least $H_k$, following the standard arguments~\cite{MotwaniR95}:
        the expected length of the phase is $k \cdot H_k$, the average cost of the algorithm for each request is $1/k$, therefore the cost of the algorithm within a~phase is at least~$H_k$.

        \item 
        Assume that the algorithm receives advice in some round of the phase.
        To receive advice, we need in expectation $1/\alpha$ rounds prior to the advice round.
        The average cost of the algorithm for each request is $1/k$, hence the expected cost is at least $\frac{1}{k \cdot \alpha}$.
    \end{enumerate}

    An optimal offline algorithm OPT incurs 1 page fault during each phase, the algorithm pays at least $\min \{H_k, \frac{1}{k \cdot \alpha}\}$, hence by summing over all phases of $\sigma$, we arrive at the desired competitive ratio.
\end{proof}

Next, we give an improved lower bound for lazy algorithms for paging.
Recall that lazy algorithms for paging are the algorithms that 
are never allowed to change its cache configuration unless there is a page miss.
This class includes RandomMark as well as most other known online paging algorithms.
Note that this definition is slightly more general than the usual definition of lazy algorithms, where the algorithm is only allowed to fetch one page per request~\cite{Borodin1998}; the intention of this definition is that the lower bound holds for metrical task systems as well.
In the classic setting without infused advice, any algorithm can be turned to a lazy algorithm without increasing its cost; note, however, that the transformed algorithm may not be randomness-oblivious.
If we restrict our attention to randomness-oblivious algorithms, the non-lazy algorithms may have an advantage over the lazy algorithms due to non-lazy algorithm's potentially frequent interaction with the oracle, which could be used by the oracle to give advice to prefetch some items even before the first cache miss occurs.

\begin{theorem}
    \label{thm:lblazy}
    Assume that an online randomized algorithm with RIA for online paging is lazy, randomness-oblivious and $c$-competitive against the oblivious adversary.
    Then $c \ge \min \{ H_k, \frac{1}{\alpha} \}$, where $\alpha$ is the infusion parameter.
\end{theorem}
\begin{proof}
    To prove the theorem, we apply Yao's Minimax Principle~\cite{Yao77} to competitiveness of randomized algorithms.
    Consider any deterministic advice-oblivious online algorithm and the probability distribution for choosing a~request sequence as in the proof of Theorem~\ref{thm:lb_paging}.

    We claim that for any advice-oblivious algorithm, the advice received in past phases cannot reduce the cost of the algorithm in future phases. 
    We argue as follows.
    First, the advice-oblivious algorithm is forbidden to store past advice in its internal memory for future use. Second, no algorithm can store meaningful advice for the future in its cache configuration: each phase contains requests to $k+1$ different items, so for any cache configuration at the start of the phase, there is always at least $1$ \emph{clean page}: a page that is requested in the phase that the algorithm does not have in the cache at the start of the phase.
       
    We will show that the expected cost of the algorithm is at least $\min\{ H_k, \frac{1}{\alpha} \}$ in any phase.
    We lower-bound the cost of the algorithm in each phase in two ways, depending on whether in this phase the algorithm receives advice in some round with a cache miss or not.
    \begin{enumerate}
        \item Assume that the algorithm does not receive advice in any round with a cache miss.
        Since the algorithm is lazy, advice received in rounds without cache misses does not influence the algorithm's cache configuration, and since the algorithm is advice-oblivious, it cannot store such advice either.
        In such case, the algorithm acts as an online algorithm without advice throughout the phase, and
        the expected cost of the algorithm in the phase is at least $H_k$, following the standard arguments~\cite{MotwaniR95}:
        the expected length of the phase is $k \cdot H_k$, the average cost of the algorithm for each request is $1/k$ because the requested page is random and each of its pages is outside the cache with equal probability, therefore the cost of the algorithm within a~phase is~$H_k$.

        \item Assume that the algorithm receives advice in a round with a cache miss.
        To receive advice at a round with a cache miss, we need in expectation $1/\alpha$ rounds with cache misses.
        Each round with a cache miss costs $1$, hence the expected cost of the algorithm is at least $1/\alpha$.
    \end{enumerate}
    An optimal offline algorithm OPT incurs a~single page fault during each phase, and the algorithm pays at least $\min \{H_k, \frac{1}{\alpha}\}$, hence by summing over all phases of $\sigma$, we arrive at the desired competitive ratio.
\end{proof}

The bound given in Theorem~\ref{thm:lblazy} is asymptotically tight for lazy algorithms. However, a~gap of a~constant factor of 2 remains. To address this gap, an optimal randomized algorithm for paging~\cite{McGeochS91} may be a possible direction for future studies.

\section{Conclusions}

We introduced a~novel method for the rigorous quantitative evaluation of online
algorithms that relaxes the worst-case perspective of classic competitive analysis. 
The infused advice model allows the seamless integration of machine-learned predictors with existing randomized online algorithms.

We leave several avenues for future research, in particular to explore the utility of our method applied to other randomized online algorithms.
Randomness-oblivious online algorithms are known for many online problems, e.g., all randomized memoryless algorithms~\cite{ChrobakL91} such as
the COINFLIP algorithm for file migration~\cite{Westbrook94} or the HARMONIC
algorithm for $k$-server~\cite{BartalG00, Motwani99} are randomness-oblivious.

\bibliographystyle{plainurl}
\bibliography{references}

\begin{thebibliography}{10}

\bibitem{AchlioptasCN00}
Dimitris Achlioptas, Marek Chrobak, and John Noga.
\newblock Competitive analysis of randomized paging algorithms.
\newblock {\em Theor. Comput. Sci.}, 234(1-2):203--218, 2000.

\bibitem{AlbersFG05}
Susanne Albers, Lene~M. Favrholdt, and Oliver Giel.
\newblock On paging with locality of reference.
\newblock {\em J. Comput. Syst. Sci.}, 70(2):145--175, 2005.

\bibitem{AlbersF18}
Susanne Albers and Dario Frascaria.
\newblock Quantifying competitiveness in paging with locality of reference.
\newblock {\em Algorithmica}, 80(12):3563--3596, 2018.

\bibitem{AlbersJ21}
Susanne Albers and Maximilian Janke.
\newblock Scheduling in the random-order model.
\newblock {\em Algorithmica}, 83(9):2803--2832, 2021.

\bibitem{AlbersKL21a}
Susanne Albers, Arindam Khan, and Leon Ladewig.
\newblock Best fit bin packing with random order revisited.
\newblock {\em Algorithmica}, 83(9):2833--2858, 2021.

\bibitem{AlbersKL21}
Susanne Albers, Arindam Khan, and Leon Ladewig.
\newblock Improved online algorithms for knapsack and {GAP} in the random order
  model.
\newblock {\em Algorithmica}, 83(6):1750--1785, 2021.

\bibitem{AlbersL16}
Susanne Albers and Sonja Lauer.
\newblock On list update with locality of reference.
\newblock {\em J. Comput. Syst. Sci.}, 82(5):627--653, 2016.

\bibitem{AlonAABN06}
Noga Alon, Baruch Awerbuch, Yossi Azar, Niv Buchbinder, and Joseph Naor.
\newblock A general approach to online network optimization problems.
\newblock {\em {ACM} Trans. Algorithms}, 2(4):640--660, 2006.

\bibitem{AlonAABN09}
Noga Alon, Baruch Awerbuch, Yossi Azar, Niv Buchbinder, and Joseph Naor.
\newblock The online set cover problem.
\newblock {\em {SIAM} J. Comput.}, 39(2):361--370, 2009.

\bibitem{ADJKR20}
Spyros Angelopoulos, Christoph D{\"{u}}rr, Shendan Jin, Shahin Kamali, and
  Marc~P. Renault.
\newblock Online computation with untrusted advice.
\newblock In {\em 11th Innovations in Theoretical Computer Science Conference,
  {ITCS} 2020 (ITCS)}, volume 151, pages 52:1--52:15, 2020.

\bibitem{AntoniadisCE0S20}
Antonios Antoniadis, Christian Coester, Marek Eli{\'{a}}s, Adam Polak, and
  Bertrand Simon.
\newblock Online metric algorithms with untrusted predictions.
\newblock In {\em Proceedings of the 37th International Conference on Machine
  Learning, {ICML} 2020, 13-18 July 2020, Virtual Event}, volume 119 of {\em
  Proceedings of Machine Learning Research}, pages 345--355. {PMLR}, 2020.

\bibitem{BamasMS20}
{\'{E}}tienne Bamas, Andreas Maggiori, and Ola Svensson.
\newblock The primal-dual method for learning augmented algorithms.
\newblock In {\em Advances in Neural Information Processing Systems 33: Annual
  Conference on Neural Information Processing Systems 2020, NeurIPS 2020},
  2020.

\bibitem{BansalCKPV22}
Nikhil Bansal, Christian Coester, Ravi Kumar, Manish Purohit, and Erik Vee.
\newblock Learning-augmented weighted paging.
\newblock In {\em Proceedings of the 2022 {ACM-SIAM} Symposium on Discrete
  Algorithms, {SODA} 2022, Virtual Conference / Alexandria, VA, USA, January 9
  - 12, 2022}, pages 67--89. {SIAM}, 2022.

\bibitem{BartalG00}
Yair Bartal and Eddie Grove.
\newblock The harmonic \emph{k}-server algorithm is competitive.
\newblock {\em J. {ACM}}, 47(1):1--15, 2000.

\bibitem{BecchettiLMSV03}
Luca Becchetti, Stefano Leonardi, Alberto Marchetti{-}Spaccamela, Guido
  Sch{\"{a}}fer, and Tjark Vredeveld.
\newblock Average case and smoothed competitive analysis of the multi-level
  feedback algorithm.
\newblock In {\em 44th Symposium on Foundations of Computer Science {(FOCS}
  2003)}, pages 462--471, 2003.

\bibitem{Belady66}
Laszlo~A. Belady.
\newblock A study of replacement algorithms for virtual-storage computer.
\newblock {\em {IBM} Syst. J.}, 5(2):78--101, 1966.

\bibitem{Ben-DavidBKTW94}
Shai Ben{-}David, Allan Borodin, Richard~M. Karp, G{\'{a}}bor Tardos, and Avi
  Wigderson.
\newblock On the power of randomization in on-line algorithms.
\newblock {\em Algorithmica}, 11(1):2--14, 1994.

\bibitem{BockenhauerKKKM17advice}
Hans{-}Joachim B{\"{o}}ckenhauer, Dennis Komm, Rastislav Kr{\'{a}}lovic,
  Richard Kr{\'{a}}lovic, and Tobias M{\"{o}}mke.
\newblock Online algorithms with advice: The tape model.
\newblock {\em Inf. Comput.}, 254:59--83, 2017.

\bibitem{Borodin1998}
Allan Borodin and Ran {El-Yaniv}.
\newblock {\em Online Computation and Competitive Analysis}.
\newblock Cambridge University Press, 1998.

\bibitem{BorodinIRS95}
Allan Borodin, Sandy Irani, Prabhakar Raghavan, and Baruch Schieber.
\newblock Competitive paging with locality of reference.
\newblock {\em J. Comput. Syst. Sci.}, 50(2):244--258, 1995.

\bibitem{BorodinLS92}
Allan Borodin, Nathan Linial, and Michael~E. Saks.
\newblock An optimal on-line algorithm for metrical task system.
\newblock {\em J. {ACM}}, 39(4):745--763, 1992.

\bibitem{BoyarFKLM2017survey}
Joan Boyar, Lene~M. Favrholdt, Christian Kudahl, Kim~S. Larsen, and Jesper~W.
  Mikkelsen.
\newblock Online algorithms with advice: A survey.
\newblock {\em ACM Comput. Surv.}, 50(2), 2017.

\bibitem{BoyarIL15}
Joan Boyar, Sandy Irani, and Kim~S. Larsen.
\newblock A comparison of performance measures for online algorithms.
\newblock {\em Algorithmica}, 72(4):969--994, 2015.

\bibitem{BuchbinderN09}
Niv Buchbinder and Joseph Naor.
\newblock The design of competitive online algorithms via a primal-dual
  approach.
\newblock {\em Found. Trends Theor. Comput. Sci.}, 3(2-3):93--263, 2009.

\bibitem{BuchbinderN2009}
Niv Buchbinder and Joseph Naor.
\newblock Online primal-dual algorithms for covering and packing.
\newblock {\em Math. Oper. Res.}, 34(2):270--286, 2009.

\bibitem{ChrobakL91}
Marek Chrobak and Lawrence~L. Larmore.
\newblock The server problem and on-line games.
\newblock In {\em On-Line Algorithms, Proceedings of a {DIMACS} Workshop},
  volume~7 of {\em {DIMACS} Series in Discrete Mathematics and Theoretical
  Computer Science}, pages 11--64. {DIMACS/AMS}, 1991.

\bibitem{CorreaCFOT21}
Jos{\'{e}}~R. Correa, Andr{\'{e}}s Cristi, Laurent Feuilloley, Tim Oosterwijk,
  and Alexandros Tsigonias{-}Dimitriadis.
\newblock The secretary problem with independent sampling.
\newblock In {\em Proceedings of the 2021 {ACM-SIAM} Symposium on Discrete
  Algorithms, {SODA}}, pages 2047--2058. {SIAM}, 2021.

\bibitem{DobrevEKKKKM17}
Stefan Dobrev, Jeff Edmonds, Dennis Komm, Rastislav Kr{\'{a}}lovic, Richard
  Kr{\'{a}}lovic, Sacha Krug, and Tobias M{\"{o}}mke.
\newblock Improved analysis of the online set cover problem with advice.
\newblock {\em Theor. Comput. Sci.}, 689:96--107, 2017.

\bibitem{DorrigivL05}
Reza Dorrigiv and Alejandro L{\'{o}}pez{-}Ortiz.
\newblock A survey of performance measures for on-line algorithms.
\newblock {\em {SIGACT} News}, 36(3):67--81, 2005.

\bibitem{EmekFKR11}
Yuval Emek, Pierre Fraigniaud, Amos Korman, and Adi Ros{\'{e}}n.
\newblock Online computation with advice.
\newblock {\em Theor. Comput. Sci.}, 412(24):2642--2656, 2011.

\bibitem{full-version}
Yuval Emek, Yuval Gil, Maciej Pacut, and Stefan Schmid.
\newblock Online algorithms with randomly infused advice.
\newblock {\em CoRR}, abs/2302.05366, 2023.
\newblock URL: \url{https://doi.org/10.48550/arXiv.2302.05366}.

\bibitem{FiatKLMSY91}
Amos Fiat, Richard~M. Karp, Michael Luby, Lyle~A. McGeoch, Daniel~Dominic
  Sleator, and Neal~E. Young.
\newblock Competitive paging algorithms.
\newblock {\em J. Algorithms}, 12(4):685--699, 1991.

\bibitem{FiatM03}
Amos Fiat and Manor Mendel.
\newblock Better algorithms for unfair metrical task systems and applications.
\newblock {\em {SIAM} J. Comput.}, 32(6):1403--1422, 2003.

\bibitem{Frederickson80}
Greg~N. Frederickson.
\newblock Probabilistic analysis for simple one- and two-dimensional bin
  packing algorithms.
\newblock {\em Inf. Process. Lett.}, 11:156--161, 1980.

\bibitem{Grove95}
Edward~F. Grove.
\newblock Online bin packing with lookahead.
\newblock In {\em Proceedings of the Sixth Annual {ACM-SIAM} Symposium on
  Discrete Algorithms}, pages 430--436. {ACM/SIAM}, 1995.

\bibitem{IraniS98}
Sandy Irani and Steven~S. Seiden.
\newblock Randomized algorithms for metrical task systems.
\newblock {\em Theor. Comput. Sci.}, 194(1-2):163--182, 1998.

\bibitem{JiangP020}
Zhihao Jiang, Debmalya Panigrahi, and Kevin Sun.
\newblock Online algorithms for weighted paging with predictions.
\newblock In {\em 47th International Colloquium on Automata, Languages, and
  Programming, {ICALP} 2020}, volume 168 of {\em LIPIcs}, pages 69:1--69:18.
  Schloss Dagstuhl - Leibniz-Zentrum f{\"{u}}r Informatik, 2020.

\bibitem{KarlinK2020survey}
Anna~R. Karlin and Elias Koutsoupias.
\newblock Beyond competitive analysis.
\newblock In Tim Roughgarden, editor, {\em Beyond the Worst-Case Analysis of
  Algorithms}, pages 529--546. Cambridge University Press, 2020.
\newblock URL: \url{https://doi.org/10.1017/9781108637435.031}.

\bibitem{KommKKM14}
Dennis Komm, Rastislav Kr{\'{a}}lovic, Richard Kr{\'{a}}lovic, and Tobias
  M{\"{o}}mke.
\newblock Randomized online algorithms with high probability guarantees.
\newblock In {\em 31st International Symposium on Theoretical Aspects of
  Computer Science (STACS)}, volume~25, pages 470--481, 2014.

\bibitem{Korman04}
Simon Korman.
\newblock On the use of randomization in the online set cover problem.
\newblock {\em Master's thesis, Weizmann Institute of Science, Rehovot,
  Israel}, 2004.

\bibitem{KoutsoupiasP00}
Elias Koutsoupias and Christos~H. Papadimitriou.
\newblock Beyond competitive analysis.
\newblock {\em {SIAM} J. Comput.}, 30(1):300--317, 2000.

\bibitem{LykourisV21}
Thodoris Lykouris and Sergei Vassilvitskii.
\newblock Competitive caching with machine learned advice.
\newblock {\em J. {ACM}}, 68(4):24:1--24:25, 2021.

\bibitem{McGeochS91}
Lyle~A. McGeoch and Daniel~Dominic Sleator.
\newblock A strongly competitive randomized paging algorithm.
\newblock {\em Algorithmica}, 6:816--825, 1991.

\bibitem{MotwaniR95}
Rajeev Motwani and Prabhakar Raghavan.
\newblock {\em Randomized Algorithms}.
\newblock Cambridge University Press, 1995.

\bibitem{Motwani99}
Rajeev Motwani and Prabhakar Raghavan.
\newblock Randomized algorithms.
\newblock In Mikhail~J. Atallah, editor, {\em Algorithms and Theory of
  Computation Handbook}, Chapman {\&} Hall/CRC Applied Algorithms and Data
  Structures series. {CRC} Press, 1999.

\bibitem{PurohitSK18}
Manish Purohit, Zoya Svitkina, and Ravi Kumar.
\newblock Improving online algorithms via {ML} predictions.
\newblock In {\em Advances in Neural Information Processing Systems 31: Annual
  Conference on Neural Information Processing Systems 2018, NeurIPS 2018},
  pages 9684--9693, 2018.

\bibitem{ReinekeS18}
Jan Reineke and Alejandro Salinger.
\newblock On the smoothness of paging algorithms.
\newblock {\em Theory Comput. Syst.}, 62(2):366--418, 2018.

\bibitem{Rivest76}
Ronald~L. Rivest.
\newblock On self-organizing sequential search heuristics.
\newblock {\em Commun. {ACM}}, 19(2):63--67, 1976.

\bibitem{Rohatgi20}
Dhruv Rohatgi.
\newblock Near-optimal bounds for online caching with machine learned advice.
\newblock In {\em Proceedings of the 2020 {ACM-SIAM} Symposium on Discrete
  Algorithms, {SODA} 2020, Salt Lake City, UT, USA, January 5-8, 2020}, pages
  1834--1845. {SIAM}, 2020.

\bibitem{Sleator1985}
Daniel~D. Sleator and Robert~E. Tarjan.
\newblock Amortized efficiency of list update and paging rules.
\newblock {\em Communications of the ACM}, 28(2):202--208, 1985.
\newblock \href {https://doi.org/10.1145/2786.2793}
  {\path{doi:10.1145/2786.2793}}.

\bibitem{Westbrook94}
Jeffery~R. Westbrook.
\newblock Randomized algorithms for multiprocessor page migration.
\newblock {\em {SIAM} J. Comput.}, 23(5):951--965, 1994.

\bibitem{Yao77}
Andrew~Chi{-}Chih Yao.
\newblock Probabilistic computations: Toward a unified measure of complexity.
\newblock In {\em 18th Annual Symposium on Foundations of Computer Science,
  Providence}, pages 222--227. {IEEE} Computer Society, 1977.

\bibitem{Young1994loose}
Neal~E. Young.
\newblock The k-server dual and loose competitiveness for paging.
\newblock {\em Algorithmica}, 11(6):525--541, 1994.
\newblock \href {https://doi.org/10.1007/BF01189992}
  {\path{doi:10.1007/BF01189992}}.

\end{thebibliography}

\end{document}